%% file: report-en.tex
\documentclass[%
 reprint,
 superscriptaddress,
 amsmath,amssymb,
 aps,
 amsthm,
 pra,
]{revtex4-1}

\usepackage{graphicx}
\usepackage{dcolumn}
\usepackage{bm}
\usepackage{braket}%
\usepackage{theorem}
\usepackage{multirow}
\usepackage{enumerate}
\usepackage{bbm}

\usepackage{color}

\theorembodyfont{\upshape}

\definecolor{gray}{RGB}{128,128,128}

\input{settings_quant.tex}

\renewcommand{\hA}{A}

\renewcommand{\hF}{F}

\renewcommand{\hP}{P}

\renewcommand{\hU}{U}
\renewcommand{\hV}{V}

\renewcommand{\hX}{X}

\renewcommand{\hOmega}{\Omega}
\renewcommand{\hPi}{\Pi}
\renewcommand{\hPhi}{\Phi}

\renewcommand{\hLambda}{\Lambda}

\renewcommand{\hrho}{\rho}
\renewcommand{\hpsi}{\psi}

\renewcommand{\ident}{\mathbbm{1}}
\newcommand{\Psipure}{\Psi_{\rm pure}}

\newcommand{\Pime}{\Pi^{(\rm e)}}
\newcommand{\hPime}{\hPi^{(\rm e)}}
\newcommand{\pime}{\pi^{(\rm e)}}
\newcommand{\Piinc}{\Pi}
\newcommand{\hPiinc}{\hPi}
\newcommand{\piinc}{\pi}
\newcommand{\mLpureme}{\mLpure^{(\rm e)}}
\newcommand{\hAme}{{\hA^{(\rm e)}}}
\newcommand{\Phime}{\Phi^{(\rm e)}}
\newcommand{\hPhime}{\hPhi^{(\rm e)}}
\newcommand{\boplus}{\diamond}
\newcommand{\s}{{(s)}}
\newcommand{\q}{{(q)}}
\renewcommand{\c}{\circ}
\newcommand{\sm}{{(1 \boplus s)}}
\newcommand{\qs}{{(q \boplus s)}}
\newcommand{\inv}[1]{\overline{#1}}
\newcommand{\mHex}{{\mH_{\rm ex}}}

\newcommand{\mHL}{{\mH_{\hLambda}}}
\newcommand{\mLpure}{\mL}
\newcommand{\mLsep}{\mL_{\rm sep}}
\newcommand{\teta}{\tilde{\eta}}

\renewcommand{\u}{{(u)}}
\newcommand{\us}{{(u \boplus s)}}

\newcommand{\tmH}{\mH_1}

\begin{document}

\preprint{APS/123-QED}

\title{Optimal quantum state discrimination with confidentiality}%

\affiliation{%
 Hitachi, Ltd., Research \& Development Group, Center for Technology Innovation - Production Engineering,
 Yokohama, Kanagawa 244-0817, Japan
}%
\affiliation{
 School of Information Science and Technology,
 Aichi Prefectural University,
 Nagakute, Aichi 480-1198, Japan
}%
\affiliation{
 Quantum Communication Research Center, Quantum ICT Research Institute, Tamagawa University,
 Machida, Tokyo 194-8610, Japan
}%
\affiliation{%
 Quantum Information Science Research Center, Quantum ICT Research Institute,
 Tamagawa University, Machida, Tokyo 194-8610, Japan
}%

\author{Kenji Nakahira}
 \affiliation{%
  Hitachi, Ltd., Research \& Development Group, Center for Technology Innovation - Production Engineering,
  Yokohama, Kanagawa 244-0817, Japan
 }%
 \affiliation{%
  Quantum Information Science Research Center, Quantum ICT Research Institute,
  Tamagawa University, Machida, Tokyo 194-8610, Japan
 }%

\author{Tsuyoshi \surname{Sasaki Usuda}}
\affiliation{
 School of Information Science and Technology,
 Aichi Prefectural University,
 Nagakute, Aichi 480-1198, Japan
}%
\affiliation{%
 Quantum Information Science Research Center, Quantum ICT Research Institute,
 Tamagawa University, Machida, Tokyo 194-8610, Japan
}%

\author{Kentaro Kato}
\affiliation{
 Quantum Communication Research Center, Quantum ICT Research Institute, Tamagawa University,
 Machida, Tokyo 194-8610, Japan
}%

\date{\today}

\begin{abstract}
 We investigate quantum state discrimination with confidentiality.
 $N$ observers share a given quantum state belonging to a finite set of known states.
 The observers want to determine the state as accurately as possible and send a discrimination result to a receiver.
 However, the observers are not allowed to get any information about which state was given.
 $N-1$ or fewer observers might try to steal the information, but if $N$ observers coexist, the honest ones
 will keep the dishonest ones from doing anything wrong.
 Assume that the state set has a certain symmetry, or more precisely, is Abelian geometrically uniform (AGU).
 We propose a protocol that realizes any optimal inconclusive measurement,
 which is a generalized version of a minimum-error measurement and an optimal unambiguous measurement,
 for any AGU state set and ensures that any combined state of $N-1$ or fewer observers
 has absolutely no information about the given state.
 Our protocol provides a method of performing a quantum measurement securely,
 which could be useful in quantum information applications.
\end{abstract}

\maketitle


\section{Introduction}

Suppose that a sender wants to send a classical message to a receiver
in the harsh environment, such as deep space, but
they cannot communicate directly.
Thus, the sender sends the message to a third party, called an observer,
and the observer sends it to the receiver.
Consider that the observer receives a quantum state $\hrho_m$ belonging to
a set of known quantum states, $\hrho_0, \hrho_1, \cdots, \hrho_{M-1}$,
which are mutually non-orthogonal.
The observer performs a quantum measurement on $\hrho_m$ and sends its result to
the receiver using classical communication.
However, the message is highly private and/or sensitive (e.g., a classified message),
and so the observer is not allowed to get any information about $m$.
What can the observer do to send as precise information about $m$ as possible to the receiver
while ensuring that the observer obtains no information?

We will show that it is possible to do this when two or more observers exist
and at least one of them is honest.
For simplicity, consider that there are two observers, Alice and Bob, and they receive the state $\hrho_m$.
We assume that Alice or Bob might try to steal the information about $m$ by illegal means,
but if they coexist, the honest one will keep the dishonest one from doing anything wrong.
We propose a protocol where they tells a discrimination result to the receiver, Charlie,
while ensuring that absolutely no information is leaked to Alice or Bob.
Let us explain our protocol using Fig.~\ref{fig:sharing_diagram}.
Alice and Bob first transform a given state $\hrho_m$ into $\hrho'_m$
as a preprocessing step, where $\hrho'_m$ is a (generally entangled) state of their composite system.
In this step, they cannot perform a wrong or evil action since they coexist.
They next independently measure their individual systems.
In this step, a dishonest observer may try to extract information about $m$.
They tell their outcomes to Charlie via classical communication.
This classical communication is encrypted to ensure that neither Alice nor Bob learns
the classical data that the other one transmits.
Charlie finally determines $m$ from their outcomes.
We refer to such a measurement as a {\it bipartite secure measurement}
if neither Alice nor Bob obtains any information about $m$ even if they act dishonestly
in their individual measurements.

\begin{figure}[tb]
 \centering
 \includegraphics[scale=0.8]{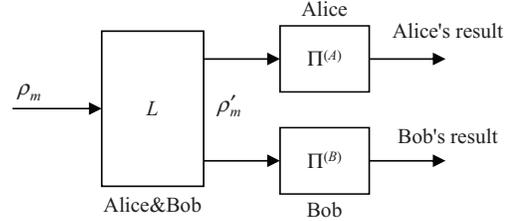}
 \caption{Data flow for a bipartite secure measurement.}
 \label{fig:sharing_diagram}
\end{figure}

As is well known, non-orthogonal states cannot be perfectly distinguished;
thus, we want to find a measurement that performs best in a certain strategy.
In one strategy, a measurement that maximizes the average correct probability
\cite{Hol-1973,Yue-Ken-Lax-1975,Hel-1976}, denoted by a minimum-error measurement,
has been investigated \cite{Bel-1975,Ban-Kur-Mom-Hir-1997,And-Bar-Gil-Hun-2002,Cho-Hsu-2003,Bar-Cro-2009}.
In another strategy, a measurement that achieves unambiguous, i.e., error-free, discrimination
with the minimum average failure probability \cite{Iva-1987,Die-1988,Per-1988},
denoted by an optimal unambiguous measurement,
has also been studied \cite{Eld-2003-unamb,Her-2007,Kle-Kam-Bru-2010,Ber-Fut-Fel-2012}.
Recently, as a more general measurement, 
a measurement that maximizes the average correct probability
with a fixed average failure probability,
which we refer to as an optimal inconclusive measurement (OIM), has been investigated
\cite{Che-Bar-1998-inc,Fiu-Jez-2003}.
Minimum-error measurements and optimal unambiguous measurements can be interpreted
as special cases of OIMs.

Remarkably, we show that any OIM, with any average failure probability,
for the state set $\{ \hrho_m \}$ can be realized with a bipartite secure measurement
if $\{ \hrho_m \}$ has certain symmetry properties,
or more precisely, if it is a (not necessarily pure) Abelian geometrically uniform (AGU) state set
\cite{Eld-For-2001,Nak-Usu-2013-group}.
Such a state set is a broad class of quantum state sets, including
phase shift keyed state sets, pulse position modulated state sets,
and linear codes with binary letter-states \cite{Eld-Meg-Ver-2004,Car-Pie-2010,Usu-Tak-Hat-Hir-1999}.
We also investigate the multipartite case
and derive that a {\it multipartite secure measurement} can realize any OIM if $\{ \hrho_m \}$
is an AGU state set.
For simplicity, throughout the main body of the paper
we consider only three linearly independent cyclic pure states.
In Appendix~\ref{append:AGU}, we will extend our technique to AGU states.

Although our scheme and secret sharing might seem somewhat similar,
they are quite different.
In classical secret sharing \cite{Bla-1979,Sha-1979},
a classical, i.e., perfectly distinguishable, secret is split among several parties.
In addition, a method for sharing an arbitrary unknown quantum state has been proposed \cite{Cle-Got-Lo-1999},
which provides a quantum version of secret sharing.
In such schemes, the parties share a classical or quantum state,
which can be perfectly reconstructed when a sufficient number of parties cooperate.
In contrast, in a bipartite secure measurement, the observers share a classical message encoded in quantum states
that are not perfectly distinguishable,
i.e., a given quantum state cannot be perfectly reconstructed from the measurement outcome.
Moreover, in our scenario, the observers cannot communicate with each other after preprocessing.
Our scheme provides a method for optimally discriminating between quantum states
with confidentiality using the basic idea of secret sharing,
though our technique is drastically different from that of secret sharing.
Note that combining quantum cryptography with classical secret sharing to
protect against eavesdropping has also been proposed \cite{Hil-1999}.

\section{Formulation}

Let us consider discrimination between $M$ quantum states
represented by density operators $\hrho_m$ $~(m \in \mI_M)$
with prior probabilities $\xi_m$, where $\mI_k = \{ 0, 1, \cdots, k-1 \}$.
$\hrho_m$ is positive semidefinite and has unit trace.
In this paper, we assume equal probabilities $\xi_m = 1/M$ for any $m \in \mI_M$.
If each $\hrho_m$ is rank one, in which case $\hrho_m$ can be expressed as
$\hrho_m = \ket{\psi_m} \bra{\psi_m}$ for any $m \in \mI_M$,
then $\Psi = \{ \hrho_m : m \in \mI_M \}$
(or $\Psipure = \{ \ket{\psi_m} : m \in \mI_M \}$) is referred to as a pure state set.
Moreover, if $\{ \ket{\psi_m} : m \in \mI_M \}$ are linearly independent,
then $\Psi$ (or $\Psipure$) is called a linearly independent pure state set.

Let $\mH$ be the space spanned by $\{ \hrho_m : m \in \mI_M \}$.
Also, let $\ident_\mK$ be the identity operator on a Hilbert space $\mK$.
If there exists a unitary operator $\hV$ on $\mH$
such that $\hV^M = \ident_\mH$ and $\hrho_m = \hV^m \hrho_0 \hV^{-m}$ for any $m \in \mI_M$,
then $\Psi$ is referred to as a cyclic state set.
In particular, if the pure state set $\Psipure$ is cyclic,
then there exists a unitary operator $\hV$ on $\mH$
such that $\hV^M = \ident_\mH$ and $\ket{\psi_m} = \hV^m \ket{\psi_0}$ for any $m \in \mI_M$
(when we choose proper global phases).
Cyclic states are special cases of AGU states.
We will give the definition of AGU states in Appendix~\ref{subsec:preliminaries}.

A quantum measurement that may return an inconclusive answer
can be described by a positive operator-valued measure (POVM)
with $M+1$ detection operators, $\Pi = \{ \hPi_m : m \in \mI_M^? \}$,
where $\mI_M^?$ is the set formed by adding element $`?'$ to the set $\mI_M$.
The detection operator $\hPi_m$ with $m \in \mI_M$ corresponds to identification of
the state $\hrho_m$, while $\hPi_?$ corresponds to the inconclusive answer.
An OIM is a measurement maximizing the average correct probability
under the constraint that the average failure probability is $p$ $~(0 \le p \le 1)$
\cite{Che-Bar-1998-inc,Eld-2003-inc,Fiu-Jez-2003};
i.e., an OIM is an optimal solution of the following problem:
\begin{eqnarray}
 \begin{array}{ll}
  {\rm maximize} & \displaystyle \sum_{m \in \mI_M} \xi_m \Tr(\hrho_m \hPi_m) \\
  {\rm subject~to} & \displaystyle \sum_{m \in \mI_M} \xi_m \Tr(\hrho_m \hPi_?) = p. \label{eq:inc_problem} \\
 \end{array}
\end{eqnarray}
A minimum error measurement is a special case of an OIM, which satisfies $p = 0$.
In this case, we can assume without loss of generality that
$\Pi$ has $M$ detection operators, i.e., $\Pi = \{ \hPi_m : m \in \mI_M \}$,
since $\Tr(\hrho_m \hPi_?) = 0$ holds for any $m \in \mI_M$.

To realize an $N$-partite secure measurement with $N \ge 2$,
we consider the procedure performed by $N$ observers and one receiver:
\begin{enumerate}[Step~1)]
 \setlength{\parskip}{0cm}
 \setlength{\itemsep}{0cm}
 \item The observers first perform an operation together to transform a given state $\hrho_m$
	   into a state of their composite system, denoted by $\hrho'_m$.
 \item Each observer independently performs a measurement on the state of his/her individual subsystem.
 \item Each observer sends his/her outcome to the receiver via private classical communication.
 \item The receiver determines $m$ from the observers' outcomes.
\end{enumerate}
Assume that in Step~1, which we call the preprocessing step,
the observers cannot perform a wrong or evil action.
In Step~2, dishonest observers may try to extract information about $m$
with wrong measurements.
Moreover, if $N \ge 3$, then $N-1$ or fewer dishonest observers may collaborate to get
the information.
We refer to a measurement according to this procedure as
an $N$-partite secure measurement if any combined state of $N-1$ or fewer observers
has absolutely no information about $m$,
or, equivalently, if any $N-1$ or fewer dishonest observers get absolutely no information about $m$.

\section{Realizing a minimum-error measurement} \label{sec:me}


Let us begin by considering a minimum-error measurement, and later extend it to an OIM.
%
\begin{thm} \label{thm:me}
 A minimum-error measurement for three linearly independent cyclic pure states
 can be realized with a bipartite secure measurement.
\end{thm}

\begin{proof}
Assume that $\Psipure = \{ \ket{\psi_m} : m \in \mI_3 \}$ is a three linearly independent pure state set.
Also, assume that $\Psipure$ is cyclic; i.e.,
there exists a unitary operator $\hV$ on $\mH$ such that
$\hV^3 = \ident_\mH$ and $\ket{\psi_m} = \hV^m \ket{\psi_0}$ for any $m \in \mI_3$.
Let $\Pime = \{ \hPime_m : m \in \mI_3 \}$ be
a POVM representing a minimum-error measurement on $\mH$.
$\Pime$ is always projective, and $\rank~\hPime_m = 1$ holds \cite{Ken-1973},
which means that $\hPime_m$ is expressed by $\hPime_m = \ket{\pime_m} \bra{\pime_m}$
with an orthonormal basis (ONB) $\{ \ket{\pime_m} : m \in \mI_3 \}$ in $\mH$.
Moreover, $\Pime$ is cyclic, i.e., $\ket{\pime_m} = \hV^m \ket{\pime_0}$ holds
for any $m \in \mI_3$ \cite{Bel-1975}.

We can see that a necessary and sufficient condition for a bipartite secure measurement is that
the candidate states after preprocessing, $\{ \hrho'_m : m \in \mI_3 \}$, satisfy
\begin{eqnarray}
 \Tr_A~\hrho'_j &=& \Tr_A~\hrho'_k, ~~~ \Tr_B~\hrho'_j = \Tr_B~\hrho'_k, \label{eq:TrBC}
\end{eqnarray}
for any $j, k \in \mI_3$, where $\Tr_A$ and $\Tr_B$, respectively, represent
the partial traces over Alice's and Bob's systems,
which implies that whether a given procedure is a bipartite secure measurement is determined only by
the preprocessing.
Indeed, suppose by contradiction that $\Tr_A~\hrho'_j \neq \Tr_A~\hrho'_k$ holds
for certain $j,k \in \mI_3$;
then, there exists Bob's measurement that gives some information to distinguish $\hrho'_j$ and $\hrho'_k$.
It follows that,
in order to obtain a bipartite secure measurement that can realize the minimum-error discrimination,
we must consider a measurement such that neither Alice nor Bob knows any information
about the outcome obtained by Charlie.
This means that Alice's measurement outcome must be independent of Charlie's outcome,
and so must that of Bob.
To realize this,
let us consider preprocessing that transforms $\ket{\pime_m}$,
corresponding to Charlie's outcome, into $\ket{\eta_m}$
such that $\Tr_A~\ket{\eta_m}\bra{\eta_m}$ and $\Tr_B~\ket{\eta_m}\bra{\eta_m}$ are independent of $m$.
To be concrete, let
\begin{eqnarray}
 \ket{\eta_m} &=& \frac{1}{\sqrt{3}} \sum_{k=0}^2 \ket{a_k} \ket{b_{m \ominus k}}, \label{eq:eta_m}
\end{eqnarray}
where $\ominus$ denotes the subtraction modulo $3$,
and $\{ \ket{a_m} : m \in \mI_3 \}$ and $\{ \ket{b_m} : m \in \mI_3 \}$ are ONBs
in Alice's and Bob's spaces, respectively.
Such preprocessing can be realized with the completely positive trace-preserving (CPTP) map
$\mLpureme(X) = \hAme X \hAme^\dagger$
($\dagger$ denotes the conjugate transpose),
where $\hAme = \sum_{k=0}^2 \ket{\eta_k} \bra{\pime_k}$.
This preprocessing turns the given state $\ket{\psi_m}$ into $\ket{\psi'_m} = \hAme \ket{\psi_m}$.
Now, we show that $\hrho'_m = \ket{\psi'_m}\bra{\psi'_m}$ satisfies Eq.~(\ref{eq:TrBC}).
Let $\chi_k = \braket{\pime_k | \psi_0}$; then,
since $\{ \ket{\pime_m} \}$ and $\{ \ket{\psi_m} \}$ are cyclic,
$\braket{\pime_k | \psi_m} = \chi_{k \ominus m}$ holds,
which gives $\ket{\psi'_m} = \sum_{k=0}^2 \chi_{k \ominus m} \ket{\eta_k}$.
Thus, from Eq.~(\ref{eq:eta_m}), we have
\begin{eqnarray}
 \Tr_A~\hrho'_m &=& \sum_{j,k,l=0}^2 \frac{\chi_{j \ominus m} \chi_{k \ominus m}^*}{3}
  \ket{b_{j \ominus l}} \bra{b_{k \ominus l}} \nonumber \\
 &=& \sum_{j',k',l'=0}^2 \frac{\chi_{j'} \chi_{k'}^*}{3} \ket{b_{j' \ominus l'}} \bra{b_{k' \ominus l'}},
  \label{eq:TrB_psim_me}
\end{eqnarray}
where $j' = j \ominus m$, $k' = k \ominus m$, and $l' = l \ominus m$.
This equation means that $\Tr_A~\hrho'_m$ is independent of $m$.
In the same way, we can easily derive that $\Tr_B~\hrho'_m$ is also independent of $m$.
Therefore, Eq.~(\ref{eq:TrBC}) holds.

The last thing we have to show is that
the minimum-error discrimination can be realized with
only the local operations to the state $\ket{\psi'_m}$.
Now, we consider the following procedure:
Alice and Bob independently perform measurements for $\ket{\psi'_m}$
in the ONBs $\{ \ket{a_n} \}$ and $\{ \ket{b_n} \}$,
and then send their outcomes $j$ and $k$, corresponding to $\ket{a_j}$ and $\ket{b_k}$, to Charlie,
respectively.
Charlie records his result as $j \oplus k$, where $\oplus$ is the addition modulo $3$.
It follows that this procedure can be represented by the POVM $\Phime = \{ \hPhime_m : m \in \mI_3 \}$ with
\begin{eqnarray}
 \hPhime_m &=& \sum_{k=0}^2 \ket{a_k}\bra{a_k} \otimes \ket{b_{m \ominus k}}\bra{b_{m \ominus k}}
  \label{eq:Phi_me}. \nonumber
\end{eqnarray}
We obtain
\begin{eqnarray}
 \braket{\psi'_m | \hPhime_n | \psi'_m} &=&
  \sum_{k=0}^2 \left| \bra{a_k}\bra{b_{n \ominus k}} \sum_{t=0}^2 \chi_{t \ominus m} \ket{\eta_t} \right|^2
  \nonumber \\
 &=& |\chi_{n \ominus m}|^2 = \braket{\psi_m | \hPime_n | \psi_m}, \label{eq:psi_Phi_psi}
\end{eqnarray}
which indicates that this procedure can realize the minimum-error discrimination.
\QED
\end{proof}

\section{Realizing an optimal inconclusive measurement (OIM)} \label{sec:OIM}

We extend the argument of Sec.~\ref{sec:me} to an OIM in the following theorem.
\begin{thm} \label{thm:inc}
 An OIM with any average failure probability
 for three linearly independent cyclic pure states can be realized with a bipartite secure measurement.
\end{thm}

\begin{proof}
Let $\Psipure = \{ \ket{\psi_m} : m \in \mI_3 \}$ be a set of three linearly independent cyclic pure states.
Also, let $\Piinc = \{ \hPiinc_m : m \in \mI_3^? \}$ be an OIM on $\mH$ for $\Psipure$.
$\hPiinc_m$ $~(m \in \mI_3)$ is rank one and thus can be expressed in the form
$\hPiinc_m = \ket{\piinc_m} \bra{\piinc_m}$ \cite{Eld-2003-inc}.
In contrast, $\hPi_?$ is generally not rank one.
Assume without loss of generality that $\Piinc$ is cyclic,
i.e., $\ket{\piinc_m} = \hV^m \ket{\piinc_0}$ holds \cite{Eld-2003-inc}.

In the proof of Theorem~\ref{thm:me},
to realize a minimum-error measurement with a bipartite secure measurement,
we exploited the fact that the POVM $\Pime$ is projective and cyclic.
We want to apply a similar approach to an OIM;
however, $\Piinc$ is generally non-projective.
We consider, instead of $\Piinc$, an OIM that is projective and cyclic.
Let $\mHex$ be a six-dimensional Hilbert space including $\mH$,
and $\Omega = \{ \Omega_m : m \in \mI_3^? \}$ be a projective measurement on $\mHex$
expressed as
\begin{eqnarray}
 \Omega_m &=& \ket{\omega_m^{(0)}} \bra{\omega_m^{(0)}}, ~~~ m \in \mI_3, \nonumber \\
 \Omega_? &=& \sum_{m=0}^2 \ket{\omega_m^{(1)}} \bra{\omega_m^{(1)}}, \nonumber
\end{eqnarray}
where $\{ \ket{\omega_m^\s} : m \in \mI_3, s \in \mI_2 \}$ is an ONB in $\mHex$.
Let $\hP$ be the orthogonal projection operator from $\mHex$ to $\mH$.
Assume that $\Omega$ is an OIM for $\Psipure$,
which satisfies $\hP \Omega_m \hP^\dagger = \hPi_m$ for any $m \in \mI_3^?$,
and that for each $s \in \mI_2$, $\{ \hP \ket{\omega_m^\s} : m \in \mI_3 \}$ is cyclic,
i.e., we have
\begin{eqnarray}
 \hP \ket{\omega_m^\s} = \hV^m \hP \ket{\omega_0^\s} \label{eq:tpi_ms}, ~~~ m \in \mI_3. \nonumber
\end{eqnarray}
As will be described later, these assumptions hold if we properly choose an ONB $\{ \ket{\omega_m^\s} \}$.
Now, under these assumptions, we show that a bipartite secure measurement
can realize an OIM for $\Psipure$.

First, we show preprocessing in which a bipartite secure measurement is possible,
i.e., Eq.~(\ref{eq:TrBC}) holds.
Consider that Alice and Bob perform the preprocessing represented by
the CPTP map $\mLpure(X) = A X A^\dagger$, where
\begin{eqnarray}
 A &=& \sum_{s=0}^1 \sum_{m=0}^2 \ket{\eta_m^\s} \bra{\omega_m^\s}, \nonumber \\
 \ket{\eta_m^\s} &=& \frac{1}{\sqrt{6}} \sum_{q=0}^1 \sum_{k=0}^2 \ket{a_k^\q}
  \ket{b_{m \ominus k}^\qs} \label{eq:eta_ms}, ~~~ s \in \mI_2, m \in \mI_3. \nonumber
\end{eqnarray}
$\boplus$ is the addition modulo $2$,
and $\{ \ket{a_m^\s} : m \in \mI_3, s \in \mI_2 \}$ and
$\{ \ket{b_m^\s} : m \in \mI_3, s \in \mI_2 \}$ are ONBs in Alice's and Bob's systems.
Since $\{ \hP \ket{\omega_m^\s} \}$ and $\{ \ket{\psi_m} \}$ are cyclic,
$\braket{\omega_k^\s | \psi_m} = \braket{\omega_{k \ominus m}^\s | \psi_0}$ holds.
Thus, we can verify that the state after preprocessing, $\ket{\psi'_m} = A\ket{\psi_m}$,
satisfies Eq.~(\ref{eq:TrBC})
in the same way as in Eq.~(\ref{eq:TrB_psim_me}).

Next, we show that an OIM can be realized with
the following procedure:
Alice and Bob independently perform the measurements for the state $\ket{\psi'_m}$
in the ONBs $\{ \ket{a_m^\s} \}$ and $\{ \ket{b_m^\s} \}$
and send their outcomes (denoted by $\ket{a_j^\q}$ and $\ket{b_k^\s}$) to Charlie, respectively.
Charlie records his result as $j \oplus k$ if $q = s$ and ``failure'' otherwise.
This procedure can be represented by the POVM $\Phi = \{ \hPhi_m : m \in \mI_3^? \}$ with
\begin{eqnarray}
 \Phi_m &=& \sum_{q=0}^1 \sum_{k=0}^2
  \ket{a_k^\q}\bra{a_k^\q} \otimes \ket{b_{m \ominus k}^\q}\bra{b_{m \ominus k}^\q},
  ~ m \in \mI_3, \nonumber \\
 \Phi_? &=& \sum_{q=0}^1 \sum_{m,k=0}^2
  \ket{a_k^\q}\bra{a_k^\q} \otimes \ket{b_{m \ominus k}^{(1 \boplus q)}}\bra{b_{m \ominus k}^{(1 \boplus q)}}.
  \nonumber
\end{eqnarray}
In a similar way to Eq.~(\ref{eq:psi_Phi_psi}), we can easily verify that
$\braket{\psi'_m | \Phi_k | \psi'_m} = \braket{\psi_m | \Omega_k | \psi_m}$.
Therefore, this procedure realizes an OIM.

Finally, we have to show that an ONB $\{ \ket{\omega_m^\s} \}$ exists such that
$\Omega$ is an OIM for $\Psipure$
and $\{ \hP \ket{\omega_m^\s} \}$ is cyclic.
Let $\{ \ket{\phi_m} : m \in \mI_3 \}$ be an ONB in $\mH$ such that
the Schatten decomposition of $\hPi_?$ is represented by
$\hPi_? = \sum_{k=0}^2 \lambda_k \ket{\phi_k} \bra{\phi_k}$.
We choose an ONB $\{ \ket{\phi_m^\s} : m \in \mI_3, s \in \mI_2 \}$ in $\mHex$
such that
\begin{eqnarray}
 \hP \ket{\phi_m^{(0)}} &=& \sqrt{1 - \lambda_m} \ket{\phi_m}, \nonumber \\
 \hP \ket{\phi_m^{(1)}} &=& \sqrt{\lambda_m} \ket{\phi_m}.
  \label{eq:inc_phi}
\end{eqnarray}
This implies that the one-dimensional subspace $\spn(\ket{\phi_m})$ of $\mH$ is associated with
the two dimensional subspace $\spn(\ket{\phi_m^{(0)}}, \ket{\phi_m^{(1)}})$ of $\mHex$.
Let $\hF_s = \sum_{k=0}^2 \ket{\phi_k^\s} \bra{\phi_k}$, which is an isometric mapping from $\mH$ to
$\spn(\ket{\phi_0^\s}, \ket{\phi_1^\s}, \ket{\phi_2^\s})$,
and $\{ \ket{\nu_m} : m \in \mI_3 \}$ be an ONB in $\mH$
satisfying $\hLambda \ket{\nu_m} = \ket{\piinc_m}$,
where $\hLambda = (\ident_\mH - \hPiinc_?)^{1/2}$ (such an ONB always exists \cite{remark_Omega}).
We choose $\ket{\omega_m^\s}$ as $\ket{\omega_m^{(0)}} = \hF_0 \ket{\nu_m}$
and $\ket{\omega_m^{(1)}} = \hF_1 \ket{\pime_m}$,
where $\ket{\pime_m}$ is a detection vector of the minimum-error measurement $\Pime$.
We show that $\{ \ket{\omega_m^\s} \}$ is an ONB in $\mHex$ that we sought.
From Eq.~(\ref{eq:inc_phi}) and the definition of $\hF_s$,
we can easily verify that $\hP \hF_0 = \hLambda$ and $\hP \hF_1 = \hPiinc_?^{1/2}$ hold.
The former equation yields
\begin{eqnarray}
 \hP \ket{\omega_m^{(0)}} = \hP \hF_0 \ket{\nu_m} = \hLambda \ket{\nu_m} = \ket{\piinc_m}, \label{eq:Ptpi_pi}
\end{eqnarray}
which indicates $\hP \Omega_m \hP^\dagger = \hPiinc_m$.
Also, we have
\begin{eqnarray}
 \hP \Omega_? \hP^\dagger &=& \hP \left( \ident_\mHex - \sum_{m=0}^2 \Omega_m \right) \hP^\dagger \nonumber \\
 &=& \ident_\mH - \sum_{m=0}^2 \hPiinc_m = \hPiinc_?, \nonumber
\end{eqnarray}
which follows from $\hP \hP^\dagger = \ident_\mH$.
Thus, $\Omega$, as well as $\Piinc$, is an OIM.
Moreover, $\{ \hP \ket{\omega_m^\s} \}$ is cyclic for each $s \in \mI_2$;
indeed, from Eq.~(\ref{eq:Ptpi_pi}),
\begin{eqnarray}
 \hP \ket{\omega_m^{(0)}} = \ket{\piinc_m} = \hV^m \ket{\piinc_0} = \hV^m \hP \ket{\omega_0^{(0)}}, \nonumber
\end{eqnarray}
and since $\hV$ commutes with $\hPiinc_?^{1/2}$,
\begin{eqnarray}
 \hP \ket{\omega_m^{(1)}} = \hPiinc_?^{1/2} \ket{\pime_m} = \hV^m \hPiinc_?^{1/2} \ket{\pime_0} = \hV^m \hP \ket{\omega_0^{(1)}}. \nonumber
\end{eqnarray}
This completes the proof.
\QED
\end{proof}

The proposed protocol for realizing an OIM is summarized as follows:
Alice and Bob first transform a given state $\ket{\psi_m}$ into $\ket{\psi'_m}$
by the preprocessing $\mLpure$.
They next perform measurements in the ONBs $\{ \ket{a_m^\s} \}$ and $\{ \ket{b_m^\s} \}$,
and send their outcomes to Charlie.
The average failure probability can be controlled by properly choosing
the ONB $\{ \ket{\omega_m^\s} \}$.
In this discussion,
we consider the preprocessing $\mL$ that transforms $\ket{\psi_m}$ into
a generally entangled state.
We can also show that a bipartite secure measurement that realizes an OIM exists
such that $\hrho'_m$ is always separable (see Appendix~\ref{subsubsec:separable}).


\section{Multipartite case}

We consider extending this scenario to the multipartite case.
The more observers there are, the more difficult it can be for dishonest observer
to steal the information about $m$.
Thus, it may be desirable to increase the number of observers
if the information must be kept highly confidential.
In the multipartite scenario, there are $N \ge 3$ observers and one receiver.
Let us consider the following protocol:
$N$ observers first share a given state by preprocessing.
After that, they independently perform measurements
and send their results to the receiver.
We refer to the measurement as an $N$-partite secure measurement if any combined state of $N-1$ or fewer observers
has absolutely no information about which state was given.
Here, we describe the case of $N = 3$ for three linearly independent cyclic pure states.
As a preprocessing step, Alice, Bob, and Charlie transform $\hrho_m$ into
$\hrho'_m = \mL(\hrho_m)$ with the CPTP map
$\mL_N(X) = A_N X A_N^\dagger$, where
\begin{eqnarray}
 A_N &=& \sum_{s=0}^1 \sum_{m=0}^2 \ket{\teta_m^\s} \bra{\omega_m^\s}, \nonumber \\
 \ket{\teta_m^\s} &=& \frac{1}{6} \sum_{q_1, q_2=0}^1 \sum_{k_1, k_2=0}^2
  \ket{a_{k_1}^{(q_1)}} \ket{b_{k_2}^{(q_2)}} \ket{c_{m \ominus k_1 \ominus k_2}^{(q_1 \boplus q_2 \boplus s)}}.
  \label{eq:eta_ms3} \nonumber
\end{eqnarray}
$\{ \ket{c_m^\s} : m \in \mI_3, s \in \mI_2 \}$ is an ONB in Charlie's system.
They independently perform measurements in the ONBs
$\{ \ket{a_m^\s} \}$, $\{ \ket{b_m^\s} \}$, and $\{ \ket{c_m^\s} \}$,
and send their outcomes (denoted by $\ket{a_j^\q}$, $\ket{b_k^\s}$, and $\ket{c_l^{(r)}}$) to Dave.
Dave records his result as $j \oplus k \oplus l$ if $q \boplus s \boplus r = 0$ and ``failure'' otherwise.
In a similar way to the bipartite case, we can see that using this procedure,
any OIM can be realized with a tripartite secure measurement.
We can show that if possible candidate states are AGU,
then any OIM can be realized with an $N$-partite secure measurement
for any $N \ge 3$ (proof in Appendix~\ref{subsec:multipartite}).

\section{Conclusion}

In summary, we have proposed a quantum measurement scheme, called an
$N$-partite secure measurement, that enables confidential communication of
classical messages via quantum channel.
In our bipartite protocol, Alice and Bob first share a quantum state obtained with preprocessing.
They next independently perform the measurements,
in which neither Alice nor Bob gets any information about which state was given
even if one does anything wrong,
and send their results to Charlie.
We stated that an OIM for any AGU state set can be realized with an $N$-partite secure measurement.

\begin{acknowledgments}
 We thank O. Hirota of Tamagawa University for the useful discussions.
 T. S. U. was supported (in part) by JSPS KAKENHI (Grant No. 24360151).
\end{acknowledgments}

\appendix

\section{Extending to Abelian geometrically uniform (AGU) states} \label{append:AGU}

In this section, we prove that any OIM for AGU states
can be realized with an $N$-partite secure measurement for any $N \ge 2$.

\subsection{Outline}

First, in Subsec.~\ref{subsec:preliminaries}, we provide definitions and a lemma.
Next, in Subsec.~\ref{subsec:bipartite}, we consider a bipartite secure measurement.
We show that any OIM for AGU states
can be realized with a bipartite secure measurement (Theorem~\ref{thm:GU_inc}).
We also show that there exists such a bipartite secure measurement
in which a given state is transformed into a separable state in a preprocessing step.
Finally, in Subsec.~\ref{subsec:multipartite}, we extend Theorem~\ref{thm:GU_inc}
to the multipartite case (Theorem~\ref{thm:GU_inc_multipartite}).

\subsection{Preliminaries} \label{subsec:preliminaries}

We consider a set of $M$ quantum states $\Psi = \{ \hrho_m : m \in \mG \}$,
where $\mG$ is an Abelian group with $M$ elements.
Let $e$ be the identity element of $\mG$.
Assume that there exists a set of $M$ operators, $\{ \hU_m : m \in \mG \}$, in which
$\hU_m$ is a unitary or anti-unitary operator on $\mH$ satisfying
\begin{eqnarray}
 \hU_e &=& \ident_\mH, \nonumber \\
 \hU_m \hU_k &=& \hU_{m \circ k}, ~~~ m, k \in \mG, \label{eq:U}
\end{eqnarray}
and
\begin{eqnarray}
 \hU_m \hrho_k \hU_m^\dagger &=& \hrho_{m \circ k}, ~~~ m,k \in \mG,
\end{eqnarray}
where $m \c k$ is the multiplication of $m$ and $k$,
and $\hU_m^\dagger$ is the operator satisfying $\hU_m^\dagger \hU_m = \hU_m \hU_m^\dagger = \ident_\mH$.
Note that if $\hU_m$ is an anti-unitary operator, then so is $\hU_m^\dagger$.
Such a state set with equal prior probabilities is called an AGU state set
\cite{Eld-For-2001,Eld-Meg-Ver-2004,Nak-Usu-2013-group}.
In particular, if $\mG$ is cyclic,
then $\Psi$ is referred to as a cyclic state set.

A quantum measurement for $\Psi$ that may return an inconclusive answer
can be described by a POVM $\Pi = \{ \hPi_m : m \in \mG_? \}$,
where $\mG_?$ is the set (which is not a group) formed by adding element $`?'$ to $\mG$.
The detection operator $\hPi_m$ with $m \in \mG$ corresponds to identification of
the state $\hrho_m$, while $\hPi_?$ corresponds to the inconclusive answer.

For a given set of unitary or anti-unitary operators $\{ \hU_m : m \in \mG \}$ on $\mH$
satisfying Eq.~(\ref{eq:U}),
we say a set of vectors $\{ \ket{a_m} : m \in \mG \}$ of $\mH$ is AGU if
\begin{eqnarray}
 \ket{a_{m \c k}} &=& \hU_m \ket{a_k}, ~~~ m,k \in \mG.
\end{eqnarray}
Let $R = \rank~\hrho_m$, which is independent of $m$ when $\Psi$ is AGU.
It follows that there exist vectors $\{ \ket{\psi_{m,r}} : m \in \mG, r \in \mI_R \}$ such that
\begin{eqnarray}
 \hrho_m &=& \sum_{r \in \mI_R} \ket{\psi_{m,r}} \bra{\hpsi_{m,r}}, ~~~ m \in \mG, \nonumber \\
 \ket{\psi_{m \c k,r}} &=& \hU_m \ket{\psi_{k,r}}, ~~~ m,k \in \mG, r \in \mI_R, \label{eq:Psi}
\end{eqnarray}
i.e., $\{ \ket{\psi_{m,r}} : m \in \mG \}$ is AGU for any $r \in \mI_R$.
Indeed, if we choose $\{ \ket{\psi_{e,r}} : r \in \mI_R \}$ such that
$\hrho_e = \sum_{r \in \mI_R} \ket{\psi_{e,r}} \bra{\hpsi_{e,r}}$
and let $\ket{\psi_{m,r}} = \hU_m \ket{\psi_{e,r}}$, then Eq.~(\ref{eq:Psi}) holds.

The following lemma shows that any OIM for AGU states
can be expressed as a projection valued measure (PVM) with certain symmetries
(proof in Appendix~\ref{append:lemma_inc}).
\begin{lemma} \label{lemma:inc}
 We consider a PVM $\Omega = \{ \hOmega_m : m \in \mG_? \}$ on
 a $2MR$-dimensional Hilbert space (denoted by $\mHex$) satisfying $\mHex \supseteq \mH$.
 Assume that $\hOmega_m$ is expressed as
 \begin{eqnarray}
  \hOmega_m &=& \sum_{r \in \mI_R} \ket{\omega_{m,r}^{(0)}}\bra{\omega_{m,r}^{(0)}}, ~~~ m \in \mG, \nonumber \\
  \hOmega_? &=& \sum_{m \in \mG} \sum_{r \in \mI_R} \ket{\omega_{m,r}^{(1)}}\bra{\omega_{m,r}^{(1)}}, \label{eq:inc_tPi}
 \end{eqnarray}
 where $\{ \ket{\omega_{m,r}^\s} : s \in \mI_2, m \in \mG, r \in \mI_R \}$ is
 an ONB in $\mHex$.
 For any $p$ with $0 \le p \le 1$, there exists an ONB $\{ \ket{\omega_{m,r}^\s} \}$
 such that
 \begin{eqnarray}
  \hP \ket{\omega_{m \c k,r}^\s} &=& \hU_m \hP \ket{\omega_{k,r}^\s},
   ~~ k,m \in \mG, s \in \mI_2, r \in \mI_R \nonumber \\ \label{eq:inc_Ppi_Upi}
 \end{eqnarray}
 holds (i.e., $\{ \hP \ket{\omega_{m,r}^\s} : m \in \mG \}$ is AGU for any $s \in \mI_2$ and $r \in \mI_R$)
 and $\Omega$ is an OIM,
 with the average failure probability of $p$, for $\Psi$.
\end{lemma}

\subsection{Bipartite secure measurement for AGU states} \label{subsec:bipartite}

\subsubsection{Realization of an OIM}

In Theorem~\ref{thm:inc}, we showed that any OIM for three linearly independent cyclic pure states
can be realized with a bipartite secure measurement.
Here, we extend this result as follows:
\begin{thm} \label{thm:GU_inc}
 An OIM, with any average failure probability, for AGU states
 can be realized with a bipartite secure measurement.
\end{thm}

\begin{proof}
Let us consider an AGU state set $\Psi = \{ \hrho_m : m \in \mG \}$.
Also, let $\Omega$ and $\{ \ket{\omega_{m,r}^\s} \}$ be the OIM
for $\Psi$ and the ONB obtained by Lemma~\ref{lemma:inc}, respectively.
Consider that two observers, Alice and Bob, perform the preprocessing represented by
the completely positive trace-preserving (CPTP) map $\mL(X) = \hA X \hA^\dagger$,
where
\begin{eqnarray}
 \hA &=& \sum_{q \in \mI_2} \sum_{k \in \mG} \sum_{r \in \mI_R}
  \ket{\eta_{k,r}^\q} \bra{\omega_{k,r}^\q}, \nonumber \\
 \ket{\eta_{k,r}^\q} &=& \frac{1}{\sqrt{2M}} \sum_{s \in \mI_2} \sum_{t \in \mG}
  \ket{a_{t,r}^\s} \ket{b_{\inv{t} \c k,r}^\qs}, \label{eq:inc_E}
\end{eqnarray}
$\inv{m}$ is the inverse element of $m$, and $\boplus$ is the addition modulo $2$.
$\{ \ket{a_{m,r}^\s} : s \in \mI_2, m \in \mG, r \in \mI_R \}$ and
$\{ \ket{b_{m,r}^\s} : s \in \mI_2, m \in \mG, r \in \mI_R \}$ are, respectively,
ONBs in Alice's and Bob's systems, each of which is $2MR$-dimensional.
Let $\ket{\psi'_{m,r}} = \hA\ket{\psi_{m,r}}$.
This map transforms $\ket{\psi_{m,r}}$ into $\ket{\psi'_{m,r}}$,
and thus $\hrho'_m = \mL(\hrho_m)$ is expressed by
\begin{eqnarray}
 \hrho'_m &=& \sum_{r \in \mI_R} \ket{\psi'_{m,r}} \bra{\hpsi'_{m,r}}.
\end{eqnarray}
We show that a bipartite secure measurement can be realized with the preprocessing $\mL$
and that an OIM can be realized with
measurements, independently performed by Alice and Bob, for the state $\hrho'_m$.

First, we show that a bipartite secure measurement can be realized.
A necessary and sufficient condition for a bipartite secure measurement is that
for any $j, k \in \mG$, $\{ \hrho'_m : m \in \mG \}$ satisfies
\begin{eqnarray}
 \Tr_A~\hrho'_j &=& \Tr_A~\hrho'_k, ~~~ \Tr_B~\hrho'_j = \Tr_B~\hrho'_k, \label{eq:TrAB}
\end{eqnarray}
where $\Tr_A$ and $\Tr_B$ are the partial traces over Alice's and Bob's systems, respectively.
Thus, it suffices to show Eq.~(\ref{eq:TrAB}).
Let $\chi_{k,r,r'}^\s = \braket{\omega_{k,r}^\s | \psi_{e,r'}} / \sqrt{2M}$;
then, we have
\begin{eqnarray}
\hspace{-1em}
 \braket{\omega_{k,r}^\s | \hpsi_{m,r'}} &=&
 \braket{\omega_{\inv{m} \c k,r}^\s | \hP^\dagger \hU_m^\dagger \hU_m | \hpsi_{e,r'}}
 \nonumber \\
 &=& \braket{\omega_{\inv{m} \c k,r}^\s | \hpsi_{e,r'}}
  = \sqrt{2M} \chi_{\inv{m} \c k,r,r'}^\s, \label{eq:inc_pi_hpsi}
\end{eqnarray}
where the first line follows from Eqs.~(\ref{eq:Psi}) and (\ref{eq:inc_Ppi_Upi})
(i.e., both $\{ \ket{\psi_{m,r}} : m \in \mG \}$ and $\{ \hP \ket{\omega_{m,r}^\s} : m \in \mG \}$ are AGU).
Thus, from Eq.~(\ref{eq:inc_E}), we obtain
\begin{eqnarray}
 \ket{\psi'_{m,r'}} &=& \hA \ket{\psi_{m,r'}}
  = \sqrt{2M} \sum_{q \in \mI_2} \sum_{k \in \mG} \sum_{r \in \mI_R}
  \chi_{\inv{m} \c k,r,r'}^\q \ket{\eta_{k,r}^\q}
  \nonumber \\
 &=& \sqrt{2M} \sum_{q \in \mI_2} \sum_{k' \in \mG} \sum_{r \in \mI_R}
  \chi_{k',r,r'}^\q \ket{\eta_{m \c k',r}^\q},
  \label{eq:inc_ket_xm}
\end{eqnarray}
where $k' = \inv{m} \c k$.
In contrast, from Eq.~(\ref{eq:inc_E}), we have
\begin{eqnarray}
 \lefteqn{ \Tr_A~\ket{\eta_{m \c k',r}^\q} \bra{\eta_{m \c l',r}^\u} } \nonumber \\
 &=&
  \frac{1}{2M} \Tr_A \left[ \sum_{s \in \mI_2} \sum_{t \in \mG}
  \ket{a_{m \c t,r}^\s} \ket{b_{\inv{t} \c k',r}^\qs} \bra{a_{m \c t,r}^\s} \bra{b_{\inv{t} \c l',r}^\us}
  \right] \nonumber \\
 &=&
  \frac{1}{2M} \sum_{s \in \mI_2} \sum_{t \in \mG}
  \ket{b_{\inv{t} \c k',r}^\qs} \bra{b_{\inv{t} \c l',r}^\us}, \label{eq:inc_ket_xm2}
\end{eqnarray}
which means that $\Tr_A~\ket{\eta_{m \c k',r}^\q} \bra{\eta_{m \c l',r}^\u}$
is independent of $m$.
Therefore, from Eq.~(\ref{eq:inc_ket_xm}),
$\Tr_A~\ket{\psi'_{m,r'}}\bra{\psi'_{m,r'}}$ is independent of $m$,
and thus so is $\Tr_A~\hrho'_m = \Tr_A \sum_{r' \in \mI_R} \ket{\psi'_{m,r'}}\bra{\psi'_{m,r'}}$.
Similarly, from
\begin{eqnarray}
 \lefteqn{ \Tr_B~\ket{\eta_{m \c k',r}^\q} \bra{\eta_{m \c l',r}^\u} } \nonumber \\
 &=&
  \frac{1}{2M} \Tr_B \left[ \sum_{s \in \mI_2} \sum_{t \in \mG}
  \ket{a_{k' \c t,r}^\qs} \ket{b_{\inv{t} \c m,r}^\s} \bra{a_{l' \c t,r}^\us} \bra{b_{\inv{t} \c m,r}^\s}
  \right] \nonumber \\
 &=&
  \frac{1}{2M} \sum_{s \in \mI_2} \sum_{t \in \mG}
  \ket{a_{k' \c t,r}^\qs} \bra{a_{l' \c t,r}^\us}, \label{eq:inc_ket_xm2_TrA}
\end{eqnarray}
we can easily derive that $\Tr_B~\hrho'_m$ is also independent of $m$.
Therefore, Eq.~(\ref{eq:TrAB}) holds.

Next, we show that an OIM can be realized with
only the local operations to the state $\hrho'_m$.
We consider the following procedure:
Alice and Bob independently perform measurements
in the ONBs $\{ \ket{a_{m,r}^\s} \}$ and $\{ \ket{b_{m,r}^\s} \}$, respectively,
and send their outcomes, denoted by $\ket{a_{t,r}^\s}$ and $\ket{b_{l,r'}^\q}$, to the receiver, Charlie.
Note that from Eq.~(\ref{eq:inc_E}), $r = r'$ always holds.
Charlie records his result as $t \c l$, which corresponds to $\ket{\psi_{t \c l}}$, if $s = q$
and ``failure'' otherwise.
It follows that this procedure can be represented by the POVM $\Pi' = \{ \hPi'_m : m \in \mG_? \}$ with
\begin{eqnarray}
\hspace{-2em}
 \hPi'_k &=& \sum_{s \in \mI_2} \sum_{t \in \mG} \sum_{r \in \mI_R} \ket{a_{t,r}^\s} \bra{a_{t,r}^\s} \otimes
  \ket{b_{\inv{t} \c k,r}^\s} \bra{b_{\inv{t} \c k,r}^\s} \label{eq:inc_Pi'k}
\end{eqnarray}
for each $k \in \mG$ and
\begin{eqnarray}
\hspace{-2em}
 \hPi'_? &=& \sum_{s \in \mI_2} \sum_{t,k \in \mG} \sum_{r \in \mI_R} \ket{a_{t,r}^\s} \bra{a_{t,r}^\s}
  \otimes \ket{b_{\inv{t} \c k,r}^\sm} \bra{b_{\inv{t} \c k,r}^\sm}. \label{eq:inc_Pi'M}
\end{eqnarray}
From Eqs.~(\ref{eq:inc_E}) and (\ref{eq:inc_Pi'k}),
for any $q, s \in \mI_2$, $k, t, l \in \mG$, and $r' \in \mI_R$, we have
\begin{eqnarray}
 \braket{\eta_{t,r}^\q | \hPi'_k | \eta_{l,r'}^\s} &=&
  \delta_{q,0} \delta_{s,0} \delta_{t,k} \delta_{l,k} \delta_{r,r'},
\end{eqnarray}
where $\delta_{a,b}$ is the Kronecker delta.
Thus, from Eq.~(\ref{eq:inc_ket_xm}),
for any $k, m \in \mG$ and $r' \in \mI_R$, we have
\begin{eqnarray}
 \braket{\psi'_{m,r'} | \hPi'_k | \psi'_{m,r'}} &=&
  2M \sum_{r \in \mI_R} |\chi_{\inv{m} \c k,r,r'}^{(0)}|^2 \nonumber \\
 &=& \braket{\psi_{m,r'} | \hOmega_k | \psi_{m,r'}},
  \label{eq:inc_psim_Pik}
\end{eqnarray}
where the second line follows from Eqs.~(\ref{eq:inc_tPi}) and (\ref{eq:inc_pi_hpsi}).
This gives
\begin{eqnarray}
 \Tr(\hrho'_m \hPi'_k) &=& \sum_{r' \in \mI_R} \braket{\psi'_{m,r'} | \hPi'_k | \psi'_{m,r'}} \nonumber \\
 &=& \sum_{r' \in \mI_R} \braket{\psi_{m,r'} | \hOmega_k | \psi_{m,r'}} = \Tr(\hrho_m \hOmega_k), \nonumber \\
 \Tr(\hrho'_m \hPi'_?) &=& 1 - \sum_{k \in \mG} \Tr(\hrho'_m \hPi'_k) \nonumber \\
 &=& 1 - \sum_{k \in \mG} \Tr(\hrho_m \hOmega_k) = \Tr(\hrho_m \hOmega_?), \label{eq:Trrho'Pi'M}
\end{eqnarray}
which indicates that the average correct and failure probabilities of $\Pi'$ for $\Psi'$
are identical to those of $\Omega$ for $\Psi$, respectively.
Therefore, this procedure realizes an OIM.
\QED
\end{proof}

\subsubsection{Preprocessing for returning a separable state} \label{subsubsec:separable}

Here, we show that an OIM can be realized with a bipartite secure measurement
even if Alice and Bob use the following CPTP map, instead of $\mL$, in the preprocessing step:
\begin{eqnarray}
 \mLsep(\hX) &=& \sum_{s \in \mI_2} \sum_{k \in \mG} \sum_{r \in \mI_R}
  \hA_{k,r}^\s \hX \hA_{k,r}^{\s \dagger}, \nonumber \\
 \hA_{k,r}^\s &=& \frac{1}{\sqrt{2M}} \ket{a_{k,r}^\s} \sum_{q \in \mI_2} \sum_{j \in \mG}
  \ket{b_{\inv{k} \c j,r}^\qs} \bra{\omega_{j,r}^\q}. \label{eq:Lsep}
\end{eqnarray}
$\mLsep$ is a CPTP map since
\begin{eqnarray}
 \lefteqn{ \sum_{s \in \mI_2} \sum_{k \in \mG} \sum_{r \in \mI_R} \hA_{k,r}^{\s \dagger} \hA_{k,r}^\s }
  \nonumber \\
 &=& \frac{1}{2M} \sum_{s,q \in \mI_2} \sum_{j,k \in \mG} \sum_{r \in \mI_R}
  \ket{\omega_{j,r}^\q} \bra{\omega_{j,r}^\q} \nonumber \\
 &=& \sum_{q \in \mI_2} \sum_{j \in \mG} \sum_{r \in \mI_R}
  \ket{\omega_{j,r}^\q} \bra{\omega_{j,r}^\q} = \ident_\mHex.
\end{eqnarray}
$\hrho'_m = \mLsep(\hrho_m)$ is a mixed state, even if $\hrho_m$ is a pure state,
but is always a separable state;
indeed, from Eq.~(\ref{eq:Lsep}), we have
\begin{eqnarray}
 \hrho'_m &=& \sum_{s \in \mI_2} \sum_{k \in \mG} \sum_{r \in \mI_R}
  \hA_{k,r}^\s \hrho_m \hA_{k,r}^{\s \dagger}, \nonumber \\
 &=& \sum_{s,q,q' \in \mI_2} \sum_{k,j,j' \in \mG} \sum_{r,t \in \mI_R}
  \chi_{\inv{m} \c j,r,t}^\q \chi_{\inv{m} \c j',r,t}^{(q') *} \nonumber \\
 & & \mbox{} \times \ket{a_{k,r}^\s} \bra{a_{k,r}^\s}
  \otimes \ket{b_{\inv{k} \c j,r}^\qs} \bra{b_{\inv{k} \c j',r}^{(q' \boplus s)}}
  \nonumber \\
 &=& \sum_{s \in \mI_2} \sum_{k,j,j' \in \mG} \sum_{r \in \mI_R} \ket{a_{k,r}^\s} \bra{a_{k,r}^\s}
  \nonumber \\
 & & \mbox{} \otimes \sum_{t \in \mI_R} \ket{\gamma_{m,k,r,t}^\s} \bra{\gamma_{m,k,r,t}^\s},
  \label{eq:sep_rho'}
\end{eqnarray}
where
\begin{eqnarray}
 \ket{\gamma_{m,k,r,t}^\s} &=& \sum_{q \in \mI_2} \sum_{j \in \mG}
  \chi_{\inv{m} \c j,r,t}^\q \ket{b_{\inv{k} \c j,r}^\qs}.
\end{eqnarray}
We show that an OIM can be realized with a bipartite secure measurement using $\mLsep$.

First, we show that a bipartite secure measurement can be realized, i.e., Eq.~(\ref{eq:TrAB}) holds.
We have
\begin{eqnarray}
 \Tr_A~\hrho'_m &=& \sum_{s,q,q' \in \mI_2} \sum_{k,j,j' \in \mG} \sum_{r,t \in \mI_R}
  \chi_{\inv{m} \c j,r,t}^\q \chi_{\inv{m} \c j',r,t}^{(q') *} \nonumber \\
 & & \mbox{} \times \ket{b_{\inv{k} \c j,r}^\qs} \bra{b_{\inv{k} \c j',r}^{(q' \boplus s)}} \nonumber \\
 &=& \sum_{s,q,q' \in \mI_2} \sum_{\kappa,\iota,\iota' \in \mG} \sum_{r,t \in \mI_R}
  \chi_{\iota,r,t}^\q \chi_{\iota',r,t}^{(q') *} \nonumber \\
 & & \mbox{} \times \ket{b_{\inv{\kappa} \c \iota,r}^\qs} \bra{b_{\inv{\kappa} \c \iota',r}^{(q' \boplus s)}},
  \nonumber \\
 \Tr_B~\hrho'_m &=& \sum_{s,q \in \mI_2} \sum_{k,j \in \mG} \sum_{r,t \in \mI_R}
  |\chi_{\inv{m} \c j,r,t}^\q|^2 \ket{a_{k,r}^\s} \bra{a_{k,r}^\s} \nonumber \\
 &=& \sum_{s,q \in \mI_2} \sum_{k,\iota \in \mG} \sum_{r,t \in \mI_R}
  |\chi_{\iota,r,t}^\q|^2 \ket{a_{k,r}^\s} \bra{a_{k,r}^\s},
\end{eqnarray}
where $\kappa = \inv{m} \c k$, $\iota = \inv{m} \c j$, and $\iota' = \inv{m} \c j'$.
Thus, $\Tr_A~\hrho'_m$ and $\Tr_B~\hrho'_m$ are independent of $m$,
i.e., Eq.~(\ref{eq:TrAB}) holds.

Next, we show that an OIM can be realized with the procedure that is the same as
in the proof of Theorem~\ref{thm:GU_inc}, except for the preprocessing step.
The procedure after preprocessing is represented
by the POVM $\Pi'$ satisfying Eqs.~(\ref{eq:inc_Pi'k}) and (\ref{eq:inc_Pi'M}).
From Eqs.~(\ref{eq:inc_Pi'k}) and (\ref{eq:sep_rho'}), we have
\begin{eqnarray}
 \hspace{-1em}
 \Tr(\hrho'_m \hPi'_k) &=&
  2M \sum_{r,t \in \mI_R} |\chi_{\inv{m} \c k,r,t}^{(0)}|^2 = \Tr(\hrho_m \hOmega_k),
\end{eqnarray}
and thus Eq.~(\ref{eq:Trrho'Pi'M}) holds.
Therefore, $\Pi'$ is an OIM for $\{ \hrho'_m = \mLsep(\hrho_m) \}$.

\subsection{Multipartite secure measurement for AGU states} \label{subsec:multipartite}

In this section, we extend Theorem~\ref{thm:GU_inc}
to a multipartite secure measurement.
\begin{thm} \label{thm:GU_inc_multipartite}
 For any $N$ with $N \ge 2$, an OIM,
 with any average failure probability, for AGU states
 can be realized with an $N$-partite secure measurement.
\end{thm}

\begin{proof}
In preparation, let us represent the preprocessing performed by $N$ observers as the CPTP map $\mL$.
Consider an AGU state set $\Psi = \{ \hrho_m : m \in \mG \}$.
Let $\Omega$ and $\{ \ket{\omega_{m,r}^\s} \}$ be the OIM for $\Psi$ and
the ONB obtained by Lemma~\ref{lemma:inc}, respectively.
Also, let $\mS_q$ be the entire set of a collection of $N$ elements of $\mI_2$,
denoted by $s = (s_0, s_1, \cdots, s_{N-1})$, satisfying $s_0 \boplus s_1 \boplus \cdots \boplus s_{N-1} = q$,
and $\mG_k$ be the entire set of a collection of $N$ elements of $\mG$,
denoted by $t = (t_0, t_1, \cdots, t_{N-1})$, satisfying $t_0 \c t_1 \c \cdots \c t_{N-1} = k$.
Assume that $\mL$ can be expressed by $\mL(X) = \hA X \hA^\dagger$ with
\begin{eqnarray}
 \hA &=& \sum_{q \in \mI_2} \sum_{k \in \mG} \sum_{r \in \mI_R}
  \ket{\eta_{k,r}^\q} \bra{\omega_{k,r}^\q}, \nonumber \\
 \ket{\eta_{k,r}^\q} &=& \frac{1}{C} \sum_{s \in \mS_q} \sum_{t \in \mG_k}
  \left( \bigotimes_{n \in \mI_N} \ket{\mu_{t_n,r}^{(s_n)}}_n \right), \label{eq:inc_mul_A}
\end{eqnarray}
and $C = (2M)^{(N-1)/2}$.
For each $n \in \mI_N$, $\{ \ket{\mu_{m,r}^\s}_n : s \in \mI_2, m \in \mG, r \in \mI_R \}$ denotes
an ONB in the $2MR$-dimensional system of the $n$-th observer.
Here, we rewrite $\ket{\eta_{k,r}^\q}$ of Eq.~(\ref{eq:inc_mul_A}) in yet another form.
Let $\mG^{N-1}$ be the entire set of a collection of $N-1$ elements of $\mG$.
For each $n,\nu \in \mI_N$, $\tau = (\tau_0, \cdots, \tau_{N-2}) \in \mG^{N-1}$, and $k \in \mG$,
let $t_{n,\nu}(\tau,k) \in \mG$ be
\begin{eqnarray}
 t_{n,\nu}(\tau,k) &=& \inv{\tau_{n-1}} \c k_{n,\nu} \c \tau_n, \label{eq:tn}
\end{eqnarray}
where $k_{n,\nu} = k$ if $n = \nu$ and $k_{n,\nu} = e$ otherwise, and $\tau_{-1} = \tau_{N-1} = e$.
For example, in the case of $N = 3$ and $\nu = 1$,
$t_{0,\nu}(\tau,k) = \tau_0$, $t_{1,\nu}(\tau,k) = \inv{\tau_0} \c k \c \tau_1$,
and $t_{2,\nu}(\tau,k) = \inv{\tau_1}$.
Here, for a fixed $\nu \in \mI_N$,
let $t = (t_{0,\nu}(\tau,k), \cdots, t_{N-1,\nu}(\tau,k))$ for a given $\tau \in \mG^{N-1}$ and $k \in \mG$;
then,
\begin{eqnarray}
 t_0 \c \cdots \c t_{N-1} &=& t_{0,\nu}(\tau,k) \c \cdots \c t_{N-1,\nu}(\tau,k) \nonumber \\
 &=& k,
\end{eqnarray}
i.e., $t \in \mG_k$, always holds from Eq.~(\ref{eq:tn}).
This implies that for a fixed $k \in \mG$, the elements of $\mG^{N-1}$, $\tau$,
are in one-to-one correspondence with the elements of $\mG_k$, $t$.
Thus, for any $\nu \in \mI_N$, $\ket{\eta_{k,r}^\q}$ can be rewritten as
\begin{eqnarray}
 \hspace{-1em}
 \ket{\eta_{k,r}^\q} &=& \frac{1}{C} \sum_{s \in \mS_q} \sum_{\tau \in \mG^{N-1}}
  \left( \bigotimes_{n \in \mI_N} \ket{\mu_{t_{n,\nu}(\tau,k),r}^{(s_n)}}_n \right). \label{eq:eta2}
\end{eqnarray}

We can prove this theorem in a similar way as Theorem~\ref{thm:GU_inc}.
We first show that an $N$-partite secure measurement can be realized with the preprocessing $\mL$
and next show that an OIM can be realized with
measurements, independently performed by $N$ observers, for the state $\hrho'_m = \mL(\hrho_m)$.

First, we show that an $N$-partite secure measurement can be realized.
It suffices to show that for any $\nu \in \mI_N$, $\{ \hrho'_m : m \in \mG \}$ satisfies
\begin{eqnarray}
 \Tr_\nu~\hrho'_j &=& \Tr_\nu~\hrho'_k, \label{eq:Trn}
\end{eqnarray}
where $\Tr_\nu$ is the partial trace over the system of the $\nu$-th observer.
Let $\chi_{k,r,r'}^\s = \braket{\omega_{k,r}^\s | \psi_{e,r'}} / C$; then we have
\begin{eqnarray}
\hspace{-1em}
 \braket{\omega_{k,r}^\s | \hpsi_{m,r'}} &=&
 \braket{\omega_{\inv{m} \c k,r}^\s | \hP^\dagger \hU_m^\dagger \hU_m | \hpsi_{e,r'}}
 \nonumber \\
 &=& \braket{\omega_{\inv{m} \c k,r}^\s | \hpsi_{e,r'}}
  = C \chi_{\inv{m} \c k,r,r'}^\s. \label{eq:inc_mul_pi_hpsi}
\end{eqnarray}
Equations (\ref{eq:inc_mul_A}) and (\ref{eq:inc_mul_pi_hpsi}) yield
\begin{eqnarray}
 \ket{\psi'_{m,r'}} &=& A \ket{\psi_{m,r'}} \nonumber \\
 &=& C \sum_{q \in \mI_2} \sum_{k \in \mG} \sum_{r \in \mI_R}
  \chi_{\inv{m} \c k,r,r'}^\q \ket{\eta_{k,r}^\q} \nonumber \\
 &=& C \sum_{q \in \mI_2} \sum_{k' \in \mG} \sum_{r \in \mI_R}
  \chi_{k',r,r'}^\q \ket{\eta_{m \c k',r}^\q}, \label{eq:inc_mul_ket_xm}
\end{eqnarray}
where $k' = \inv{m} \c k$.
In contrast, Eq.~(\ref{eq:eta2}) yields
\begin{eqnarray}
 \lefteqn{ \Tr_\nu \ket{\eta_{m \c k',r}^\q} \bra{\eta_{m \c l',r}^\u} } \nonumber \\
 &=& \frac{1}{C^2} \Tr_\nu \Biggl[ \sum_{s \in \mS_q} \sum_{\tau \in \mG^{N-1}}
  \left( \bigotimes_{n \in \mI_N} \ket{\mu_{t_{n,\nu}(\tau,m \c k'),r}^{(s_n)}}_n \right) \nonumber \\
 & & \times \sum_{s' \in \mS_u} \sum_{\tau' \in \mG^{N-1}}
  \left( \bigotimes_{n \in \mI_N} \bra{\mu_{t_{n,\nu}(\tau',m \c l'),r}^{(s'_n)}}_n \right) \Biggr] \nonumber \\
 &=& \frac{1}{C^2} \sum_{s \in \mS_q} \sum_{s' \in \mS_u} \sum_{\tau, \tau' \in \mG^{N-1}}
  \epsilon(\tau,\tau',m \c k',m \c l') \nonumber \\
 & & \times \left( \bigotimes_{\substack{n \in \mI_N \\ n \neq \nu}} \ket{\mu_{t_{n,\nu}(\tau,m \c k'),r}^{(s_n)}}
			\bra{\mu_{t_{n,\nu}(\tau',m \c l'),r}^{(s'_n)}}_n \right), \label{Tr_nu_eta}
\end{eqnarray}
where
\begin{eqnarray}
 \epsilon(\tau,\tau',m \c k',m \c l') &=&
  \left\{
   \begin{array}{ll}
	1, & ~ t_{\nu,\nu}(\tau,k') = t_{\nu,\nu}(\tau',l'), \\
	0, & ~ \mbox{otherwise}. \\
   \end{array} \right. \nonumber \\
 \label{eq:epsilon}
\end{eqnarray}
We should note that $t_{\nu,\nu}(\tau,k') = t_{\nu,\nu}(\tau',l')$ in Eq.~(\ref{eq:epsilon})
is equivalent to $t_{\nu,\nu}(\tau,m \c k') = t_{\nu,\nu}(\tau',m \c l')$, since $\mG$ is Abelian.
Since $t_{n,\nu}(\tau,m \c k')$ is independent of $m$ whenever $n \neq \nu$, which follows from Eq.~(\ref{eq:tn}),
and $\epsilon(\tau,\tau',m \c k',m \c l')$ is also independent of $m$,
the right-hand side of the last equality of Eq.~(\ref{Tr_nu_eta}) is independent of $m$.
Thus, from Eq.~(\ref{eq:inc_mul_ket_xm}),
$\Tr_\nu \ket{\psi'_{m,r'}}\bra{\psi'_{m,r'}}$ is independent of $m$.
Therefore, $\Tr_\nu~\hrho'_m = \Tr_\nu \sum_{r' \in \mI_R} \ket{\psi'_{m,r'}}\bra{\psi'_{m,r'}}$
is also independent of $m$,
which means that Eq.~(\ref{eq:Trn}) holds for any $\nu \in \mI_N$.

Next, we show that an OIM can be realized with
measurements, independently performed by $N$ observers, for the state $\hrho'_m$.
We consider the following procedure:
For each $n \in \mI_N$, the $n$-th observer independently performs the measurement
in the ONB $\{ \ket{\mu_{t_n,r_n}^{(s_n)}}_n : s_n \in \mI_2, t_n \in \mG, r_n \in \mI_R \}$
and sends his/her outcome (denoted by $\ket{\mu_{t_n,r_n}^{(s_n)}}_n$) to the receiver
(note that from Eq.~(\ref{eq:inc_mul_A}), $r_0 = r_1 = \cdots = r_{N-1}$ always holds).
Let $k = t_0 \c t_1 \c \cdots \c t_{N-1}$.
The receiver records his/her result as $k$, which corresponds to $\ket{\psi_k}$,
if $s_0 \boplus s_1 \boplus \cdots \boplus s_{N-1} = 0$ and ``failure'' otherwise.
This procedure can be represented by the POVM $\Pi' = \{ \hPi'_m : m \in \mG_? \}$ with
\begin{eqnarray}
 \hPi'_k &=& \sum_{s \in \mS_0} \sum_{t \in \mG_k} \sum_{r \in \mI_R}
 \left( \bigotimes_{n \in \mI_N} \ket{\mu_{t_n,r}^{(s_n)}} \bra{\mu_{t_n,r}^{(s_n)}}_n \right)
 \label{eq:inc_mul_Pi'k}
\end{eqnarray}
for each $k \in \mG$ and
\begin{eqnarray}
 \hspace{-2em}
 \hPi'_? &=& \sum_{s \in \mS_1} \sum_{k \in \mG} \sum_{t \in \mG_k} \sum_{r \in \mI_R}
 \left( \bigotimes_{n \in \mI_N} \ket{\mu_{t_n,r}^{(s_n)}} \bra{\mu_{t_n,r}^{(s_n)}}_n \right).
 \label{eq:inc_mul_Pi'M}
\end{eqnarray}
From Eqs.~(\ref{eq:inc_mul_A}) and (\ref{eq:inc_mul_Pi'k}),
for any $q, s \in \mI_2$, $k, t, l \in \mG$, and $r' \in \mI_R$, we have
\begin{eqnarray}
 \braket{\eta_{t,r}^\q | \hPi'_k | \eta_{l,r'}^\s} &=&
  \delta_{q,0} \delta_{s,0} \delta_{t,k} \delta_{l,k} \delta_{r,r'},
\end{eqnarray}
and thus, from Eq.~(\ref{eq:inc_mul_ket_xm}),
for any $k, m \in \mG$ and $r' \in \mI_R$, we have
\begin{eqnarray}
 \braket{\psi'_{m,r'} | \hPi'_k | \psi'_{m,r'}} &=&
  C^2 \sum_{r \in \mI_R} |\chi_{\inv{m} \c k,r,r'}^{(0)}|^2 \nonumber \\
 &=& \braket{\psi_{m,r'} | \hOmega_k | \psi_{m,r'}},
\end{eqnarray}
where the second line follows from Eq.~(\ref{eq:inc_mul_pi_hpsi}).
Thus, we have
\begin{eqnarray}
 \Tr(\hrho'_m \hPi'_k) &=& \sum_{r' \in \mI_R} \braket{\psi'_{m,r'} | \hPi'_k | \psi'_{m,r'}} \nonumber \\
 &=& \sum_{r' \in \mI_R} \braket{\psi_{m,r'} | \hOmega_k | \psi_{m,r'}} = \Tr(\hrho_m \hOmega_k), \nonumber \\
 \Tr(\hrho'_m \hPi'_?) &=& 1 - \sum_{k \in \mG} \Tr(\hrho'_m \hPi'_k) \nonumber \\
 &=& 1 - \sum_{k \in \mG} \Tr(\hrho_m \hOmega_k) = \Tr(\hrho_m \hOmega_?),
\end{eqnarray}
which indicates that the average correct and failure probabilities of $\Pi'$ for $\Psi'$
are identical to those of $\Omega$ for $\Psi$, respectively.
Therefore, this procedure realizes an OIM.
\QED
\end{proof}

\section{Proof of Lemma~\ref{lemma:inc}} \label{append:lemma_inc}

\subsection{Preparations}

Before we provide the proof, we provide definitions and facts.
Let $\Pi = \{ \hPi_m : m \in \mG_? \}$ be an OIM on $\mH$,
with the average failure probability of $p$, for $\Psi$.
From Refs.~\cite{Eld-2003-inc,Nak-Usu-2013-group},
we can assume without loss of generality that
\begin{eqnarray}
 \hPi_{m \c k} &=& \hU_m \hPi_k \hU_m^\dagger, ~~~ m,k \in \mG, \nonumber \\
 \hPi_? &=& \hU_m \hPi_? \hU_m^\dagger, ~~~ m \in \mG. \label{eq:inc_Pi}
\end{eqnarray}
We can choose proper vectors $\{ \ket{\pi_{m,r}} : m \in \mG, r \in \mI_R \}$ such that
\begin{eqnarray}
 \hPi_m &=& \sum_{r \in \mI_R} \ket{\pi_{m,r}} \bra{\pi_{m,r}}, ~~~ m \in \mG, \nonumber \\
 \ket{\pi_{m \c k,r}} &=& \hU_m \ket{\pi_{k,r}}, ~~~ m,k \in \mG. \label{eq:Pi_pi}
\end{eqnarray}
Indeed, if we choose $\{ \ket{\pi_{e,r}} : r \in \mI_R \}$ such that
$\hPi_e = \sum_{r \in \mI_R} \ket{\pi_{e,r}} \bra{\pi_{e,r}}$
and let $\ket{\pi_{m,r}} = \hU_m \ket{\pi_{e,r}}$, then Eq.~(\ref{eq:Pi_pi}) holds.
Let $D = \dim~\mH$ and $\tmH$ be an $MR$-dimensional Hilbert space satisfying $\mH \subseteq \tmH \subseteq \mHex$.
$\tmH$ always exists since $D \le MR < 2MR$ holds, and $\tmH = \mH$ obviously holds if $D = MR$.
Also, let $\hLambda = (\ident_\mH - \hPi_?)^{1/2}$; then
the Schatten decomposition of $\hLambda$ can be represented by
\begin{eqnarray}
 \hLambda &=& \sum_{d \in \mI_{MR}} \sqrt{1 - \lambda_d} \ket{\phi_d} \bra{\phi_d}, \label{eq:inc_PiM}
\end{eqnarray}
where $\{ \ket{\phi_d} : d \in \mI_{MR} \}$ is an ONB in $\tmH$.
Let $\mHL = \supp~\hLambda$.
Since $\mHL \subseteq \mH$ holds, we assume without loss of generality that
$\{ \ket{\phi_d} : d \in \mI_D \}$ spans the space $\mH$,
and thus, $\lambda_d = 1$ holds for any $d$ with $D \le d < MR$.
$\hP$ is expressed by
\begin{eqnarray}
 \hP &=& \sum_{d \in \mI_D} \ket{\phi_d} \bra{\phi_d}. \label{eq:P}
\end{eqnarray}

Now, we show that there exist two ONBs $\{ \ket{v_{m,r}^{(0)}} : m \in \mG, r \in \mI_R \}$ and
$\{ \ket{v_{m,r}^{(1)}} : m \in \mG, r \in \mI_R \}$ in $\tmH$ such that
\begin{eqnarray}
 \hLambda \ket{v_{m,r}^{(0)}} &=& \ket{\pi_{m,r}}, ~~~ m \in \mG, r \in \mI_R, \label{eq:Lomega_pi} \\
 \hP \ket{v_{m \c k,r}^{(1)}} &=& \hU_m \hP \ket{v_{k,r}^{(1)}}, ~~~ m,k \in \mG, r \in \mI_R.
  \label{eq:omega_Uomega}
\end{eqnarray}
First, we show an ONB $\{ \ket{v_{m,r}^{(0)}} \}$ satisfying Eq.~(\ref{eq:Lomega_pi}).
Let $\ket{\pi'_{m,r}} = \hLambda^+ \ket{\pi_{m,r}}$
($\hLambda^+$ is the Moore-Penrose inverse of $\hLambda$); then from Eq.~(\ref{eq:Pi_pi}), we obtain
\begin{eqnarray}
 \sum_{m \in \mG} \sum_{r \in \mI_R} \ket{\pi'_{m,r}}\bra{\pi'_{m,r}}
  &=& \hLambda^+ \left( \sum_{m \in \mG} \hPi_m \right) \hLambda^+ \nonumber \\
 &=& \hLambda^+ \hLambda^2 \hLambda^+ = \ident_\mHL.
\end{eqnarray}
Thus, $\{ \ket{\pi'_{m,r}}\bra{\pi'_{m,r}} : m \in \mG, r \in \mI_R \}$ is a POVM on $\mHL$.
From Naimark's theorem, there exists an ONB, denoted by $\{ \ket{v_{m,r}^{(0)}} \}$, in $\tmH$
such that $\hP_\hLambda \ket{v_{m,r}^{(0)}} = \ket{\pi'_{m,r}}$ \cite{Akh-Gla-1966},
where $\hP_\hLambda$ is the orthogonal projection operator from $\mH_1$ to $\mHL$.
It follows that, for any $m \in \mG$ and $r \in \mI_R$,
\begin{eqnarray}
 \hLambda \ket{v_{m,r}^{(0)}} &=& \hLambda \ket{\pi'_{m,r}} = \ket{\pi_{m,r}},
\end{eqnarray}
i.e., Eq.~(\ref{eq:Lomega_pi}) holds.
Next, we show an ONB $\{ \ket{v_{m,r}^{(1)}} \}$ satisfying Eq.~(\ref{eq:omega_Uomega}).
Let $\Pime = \{ \hPime_m : m \in \mG \}$ be a minimum-error measurement on $\mH$ for $\Psi$.
Let us choose $\ket{\pime_{m,r}}$ such that
\begin{eqnarray}
 \hPime_m &=& \sum_{r \in \mI_R} \ket{\pime_{m,r}} \bra{\pime_{m,r}}, ~~~ m \in \mG, \nonumber \\
 \ket{\pime_{m \c k,r}} &=& \hU_m \ket{\pime_{k,r}}, ~~~ m,k \in \mG.
\end{eqnarray}
In the same manner as in Eq.~(\ref{eq:Pi_pi}), we can always choose such $\ket{\pime_{m,r}}$.
From Naimark's theorem, there exists an ONB, denoted by $\{ \ket{v_{m,r}^{(1)}} \}$, in $\tmH$
such that $\hP \ket{v_{m,r}^{(1)}} = \ket{\pime_{m,r}}$.
Thus,
\begin{eqnarray}
 \hP \ket{v_{m \c k,r}^{(1)}} &=& \ket{\pime_{m \c k,r}} = \hU_m \ket{\pime_{k,r}}
  = \hU_m \hP \ket{v_{k,r}^{(1)}}, \nonumber \\
\end{eqnarray}
i.e., Eq.~(\ref{eq:omega_Uomega}) holds.

\subsection{Derivation of $\{ \ket{\omega_{m,r}^\s} \}$}

Let us remind that $\mHex$ is a $2MR$-dimensional Hilbert space satisfying $\mHex \supseteq \tmH$.
We choose an ONB $\{ \ket{\phi_d^\s} : s \in \mI_2, d \in \mI_{MR} \}$ in $\mHex$ such that
for any $d \in \mI_{MR}$,
\begin{eqnarray}
 \hP_1 \ket{\phi_d^{(0)}} &=& \sqrt{1 - \lambda_d} \ket{\phi_d}, \nonumber \\
 \hP_1 \ket{\phi_d^{(1)}} &=& \sqrt{\lambda_d} \ket{\phi_d},
  \label{eq:inc_phi2}
\end{eqnarray}
where $\hP_1$ is the orthogonal projection operator from $\mHex$ to $\mH_1$.
This implies that the one-dimensional subspace $\spn(\ket{\phi_d})$ of $\tmH$ is associated with
the two-dimensional subspace $\spn(\ket{\phi_d^{(0)}}, \ket{\phi_d^{(1)}})$ of $\mHex$.
We define $\ket{\omega_{m,r}^\s}$ as
\begin{eqnarray}
 \ket{\omega_{m,r}^\s} &=& \hF_s \ket{v_{m,r}^\s}, ~~~ s \in \mI_2, m \in \mG, r \in \mI_R, \nonumber \\
 \hF_s &=& \sum_{d \in \mI_{MR}} \ket{\phi_d^\s}\bra{\phi_d}, ~~~ s \in \mI_2.
  \label{eq:inc_Fs}
\end{eqnarray}
We can easily verify that $\{ \ket{\omega_{m,r}^\s} : s \in \mI_2, m \in \mG, r \in \mI_R \}$
is an ONB in $\mHex$.
Indeed, we obtain
\begin{eqnarray}
 \braket{\omega_{m,r}^\s | \omega_{m',r'}^\s} &=& \braket{v_{m,r}^\s | v_{m',r'}^\s} = \delta_{m,m'} \delta_{r,r'},
\end{eqnarray}
where $\delta_{a,b}$ is the Kronecker delta, which follows from $\hF_s$ being an isometric mapping,
and $\braket{\omega_{m,r}^{(0)} | \omega_{m',r'}^{(1)}} = 0$ holds from $\hF_0^\dagger \hF_1 = 0$.

Let us consider the PVM $\Omega$ defined by Eq.~(\ref{eq:inc_tPi}) with this ONB $\{ \ket{\omega_{m,r}^\s} \}$.
We will prove that $\Omega$ is an OIM, with the average failure probability of $p$,
for $\Psi$ and that Eq.~(\ref{eq:inc_Ppi_Upi}) holds.

\subsection{Proof of $\Omega$ to be an OIM}

Here, we show that $\Omega$ is an OIM, with the average failure probability of $p$,
for $\Psi$.
From Eqs.~(\ref{eq:inc_PiM}), (\ref{eq:P}), (\ref{eq:inc_phi2}), and (\ref{eq:inc_Fs}), we have
\begin{eqnarray}
 \hP \hF_0 &=& \sum_{d \in \mI_D} \sum_{j \in \mI_{MR}} \ket{\phi_d} \braket{\phi_d | \phi_j^{(0)}} \bra{\phi_j}
  \nonumber \\
 &=& \sum_{d \in \mI_D} \sqrt{1 - \lambda_d} \ket{\phi_d} \bra{\phi_d} = \hLambda, \nonumber \\
 \hP \hF_1 &=& \sum_{d \in \mI_D} \sqrt{\lambda_d} \ket{\phi_d} \bra{\phi_d}
 = (\ident_\mH - \hLambda^2)^{1/2} = \hPi_?^{1/2}. \nonumber \\
 \label{eq:inc_PFs}
\end{eqnarray}
From Eqs.~(\ref{eq:Lomega_pi}), (\ref{eq:inc_Fs}), and (\ref{eq:inc_PFs}),
for any $m \in \mG$ and $r \in \mI_R$, we have
\begin{eqnarray}
 \hP \ket{\omega_{m,r}^{(0)}} &=& \hLambda \ket{v_{m,r}^{(0)}} = \ket{\pi_{m,r}}, \label{eq:Ppi_pi}
\end{eqnarray}
which gives $\hP \hOmega_m \hP^\dagger = \hPi_m$ for any $m \in \mG$.
Also, we have
\begin{eqnarray}
 \hP \hOmega_? \hP^\dagger &=& \hP \left( \ident_\mHex - \sum_{m \in \mG} \hOmega_m \right) \hP^\dagger \nonumber \\
 &=& \ident_\mH - \sum_{m \in \mG} \hPi_m ~=~ \hPi_?.
  \label{eq:inc_PPi}
\end{eqnarray}
Therefore, $\Omega$ is an optimal solution of problem~(\ref{eq:inc_problem}) as well as $\Pi$, i.e.,
$\Omega$ is an OIM, with the average failure probability of $p$, for $\Psi$.

\subsection{Proof of Eq.~(\ref{eq:inc_Ppi_Upi})}

Here, we show Eq.~(\ref{eq:inc_Ppi_Upi}).
From Eqs.~(\ref{eq:Pi_pi}) and (\ref{eq:Ppi_pi}),
for any $m,k \in \mG$ and $r \in \mI_R$, we have
\begin{eqnarray}
 \hP \ket{\omega_{m \c k,r}^{(0)}} &=& \ket{\pi_{m \c k,r}} = \hU_m \ket{\pi_{k,r}}
  = \hU_m \hP \ket{\omega_{k,r}^{(0)}}. \nonumber \\
\end{eqnarray}
Moreover, from Eqs.~(\ref{eq:omega_Uomega}), (\ref{eq:inc_Fs}), and (\ref{eq:inc_PFs}),
for any $m,k \in \mG, r \in \mI_R$, we have
\begin{eqnarray}
 \hP \ket{\omega_{m \c k,r}^{(1)}} &=& \hPi_?^{1/2} \ket{v_{m \c k,r}^{(1)}}
  = \hPi_?^{1/2} \hP \ket{v_{m \c k,r}^{(1)}} \nonumber \\
 &=& \hPi_?^{1/2} \hU_m \hP \ket{v_{k,r}^{(1)}}
  = \hU_m \hPi_?^{1/2} \ket{v_{k,r}^{(1)}} \nonumber \\
 &=& \hU_m \hP \ket{\omega_{k,r}^{(1)}},
\end{eqnarray}
where the fourth equality follows since $\hPi_?^{1/2}$ commutes with $\hU_m$ from Eq.~(\ref{eq:inc_Pi}).
Therefore, Eq.~(\ref{eq:inc_Ppi_Upi}) holds.
\QED

\input{report-en.bbl}

\end{document}

%% file: settings_quant.tex
\newtheorem{thm}{Theorem}

\newtheorem{lemma}[thm]{Lemma}

\newtheorem{proof}{Proof}

\newcommand{\hA}{\hat{A}}

\newcommand{\hF}{\hat{F}}

\newcommand{\hP}{\hat{P}}

\newcommand{\hU}{\hat{U}}
\newcommand{\hV}{\hat{V}}

\newcommand{\hX}{\hat{X}}

\newcommand{\hOmega}{\hat{\Omega}}
\newcommand{\hPi}{\hat{\Pi}}
\newcommand{\hPhi}{\hat{\Phi}}

\newcommand{\hLambda}{\hat{\Lambda}}

\newcommand{\hrho}{\hat{\rho}}
\newcommand{\hpsi}{\hat{\psi}}

\newcommand{\mG}{\mathcal{G}}
\newcommand{\mH}{\mathcal{H}}
\newcommand{\mI}{\mathcal{I}}

\newcommand{\mK}{\mathcal{K}}
\newcommand{\mL}{\mathcal{L}}

\newcommand{\mS}{\mathcal{S}}

\newcommand{\ident}{\hat{1}}

\newcommand{\QED}{\hspace*{0pt}\hfill $\blacksquare$}

\newcommand{\Tr}{{\rm Tr}}
\newcommand{\rank}{{\rm rank}}
\newcommand{\supp}{{\rm supp}}

\newcommand{\spn}{{\rm span}}

\def\gauss_sym#1{{\lfloor #1 \rfloor}}

%% file: report-en.bbl
%

%% file: report-en.bbl
\begin{thebibliography}{30}
\bibitem [{rem()}]{remark_Omega}%
  \BibitemOpen
  \href@noop {} {}\bibinfo {note} {Let $\ket{\pi'_m} = \hLambda^+ \ket{\pi_m}$,
  where $\hLambda^+$ is the Moore-Penrose inverse of $\hLambda$. Since
  $\sum_{m=0}^2 \ket{\pi'_m}\bra{\pi'_m} = \hLambda^+ \hLambda^2 \hLambda^+ =
  \ident_\mHL$ ($\mHL$ is the support of $\hLambda$), $\{
  \ket{\pi'_m}\bra{\pi'_m} : m \in \mI_3 \}$ is a POVM. From Naimark's theorem,
  an ONB $\{ \ket{\nu_m} : m \in \mI_3 \}$ in $\mH$ exists such that the
  orthogonal projection of $\ket{\nu_m}$ onto $\mHL$ is $\ket{\pi'_m}$. This
  ONB satisfies $\hLambda \ket{\nu_m} = \hLambda \ket{\pi'_m} = \ket{\pi_m}$
  (note that $\ket{\pi_m} \in \supp~\hLambda$ since $\hLambda^2 = \ident_\mH -
  \hPiinc_? = \sum_{m=0}^2 \ket{\pi_m}\bra{\pi_m}$).}\BibitemShut {Stop}%
\end{thebibliography}
